\documentclass[letterpaper]{llncs}

\usepackage{url}
\usepackage{latexsym}
\usepackage{amsmath}
\usepackage{amssymb}
\usepackage{theorem}
\usepackage{subfigure}
\usepackage{graphicx}
\usepackage{multirow}
\usepackage{wrapfig}
\usepackage[dvipsnames,usenames]{color}
\usepackage[ruled,noend,nofillcomment,linesnumbered,titlenotnumbered]{algorithm2e}

\newcommand{\prt}[2]{{#1}/{#2}} 

\renewcommand{\O}{\mathcal{O}}
\newcommand{\remove}{\mathit{Remove}}
\renewcommand{\sim}{\mathit{Sim}}
\newcommand{\tr}{\delta}
\newcommand{\labs}{\Sigma}

\newcommand{\wordsim}{{\sim_\init}}

\newcommand{\wordsimrel}{\rho}
\newcommand{\pre}{\mathrm{\delta}^{-1}}
\newcommand{\post}{\mathrm{\delta}}
\newcommand{\init}{\mathit{I}}
\newcommand{\inp}[1]{\mathrm{in}(#1)}
\newcommand{\out}[1]{\mathrm{out}(#1)}
\newcommand{\ass}{\leftarrow}
\newcommand{\prev}{\mathsf{prev}}

\newcommand{\outpre}{\mathit{Out}}

\newcommand{\rel}{\mathit{Rel}} 
\newcommand{\Split}{\mathit{Split}}

\newcommand{\lara}[1]{\langle #1 \rangle}
\renewcommand{\above}[2]{#1(#2)}
\newcommand{\bellow}[2]{#1^{-1}(#2)}
\newcommand{\cnt}{\mathit{Count}}
\newcommand{\trrelset}{\{\delta_a \mid a\in\labs\}}
\newcommand{\rankof}{r}
\newcommand{\maxrank}{r_\mathsf{m}}
\newcommand{\nat}{\mathbb{N}}
\renewcommand{\part}[2]{{#1}/{#2}}

\newcommand{\us}{U} 
\newcommand{\lpre}{D} 
\newcommand{\vect}[2]{(#1_1 \dotsc #1_{#2})}
\newcommand{\trans}[4]{(#1_1\ldots #1_#2,#3,#4)}
\newcommand{\downsim}{\mathit{D}}
\newcommand{\upsim}{\mathit{U}}
\newcommand{\lhsof}{{\it Lhs}}
\newcommand{\hole}{\square}
\newcommand{\envof}{Env}
\newcommand{\Aref}{}

\newcommand{\Algo}{
\NoCaptionOfAlgo
\caption{\small LRT Algorithm}
\KwIn{an LTS $T = (S,\labs,\trrelset)$, partition-relation pair
$\lara{P_\init,\rel_\init}$
}
\KwOut{partition-relation pair $\lara{P,\rel}$}

\BlankLine

\tcc{initialization}
$\lara{P,\rel} \ass \lara{P_I,\rel_I}$\tcc*[f]{$\ass \lara{P_{\init\cap\outpre},\rel_{\init\cap\outpre}}$}

\ForAll(\tcc*[f]{$a\in\inp{B}$}){$B\in P$ \KwSty{and}  $a\in\labs$}{
	\lForAll(\tcc*[f]{$v\in\pre_a(S)$}){$v\in S$}{
		$\cnt_a(v,B) = |\post_a(v)\cap\bigcup\above\rel B|$
	}
	$\remove_a(B)\ass S\setminus\pre_a(\bigcup\above \rel B)$\tcc*[f]{$\ass \pre_a(S){\setminus}\pre_a(\bigcup\above \rel B)$}
}

\BlankLine
\tcc{computation}
\While{\KwSty{exists} $B\in P$ \KwSty{and} $a\in\labs$ \KwSty{such that} $\remove_a(B)\neq\emptyset$}{
        $\remove\ass\remove_a(B)$\;$\remove_a(B)\ass\emptyset$\;
        $\lara{P_{\prev},\rel_\prev}\ass \lara{P,\rel}$\;
        $P\ass\Split(P,\remove)$\;
	$\rel \ass \{(C,D)\in P\times P\mid (C_\prev,D_\prev)\in\rel_\prev\}$\; 
        \ForAll(\tcc*[f]{$b\in\inp{C}$}){$C\in P$ \KwSty{and} $b\in\labs$}{
		$\remove_b(C)\ass \remove_b(C_\prev)$\;
		\lForAll(\tcc*[f]{$v\in\pre_b(S)$}){$v\in S$}{
			$\cnt_b(v,C)\ass \cnt_b(v,C_\prev)$
		}
        }
        \ForAll{$C\in P$ \KwSty{such that} $C\cap\pre_a(B)\neq\emptyset$}{
                \ForAll{$D\in P$ \KwSty{such that} $D\subseteq \remove$}{
			\If{$(C,D)\in\rel$}{
				$\rel\ass\rel\setminus\{(C,D)\}$\;
				\ForAll(\tcc*[f]{$b\in \inp D\cap\inp{C}$}){$b\in\labs$ \KwSty{and} $v\in \pre_b(D)$}{
					$\cnt_b(v,C)\ass\cnt_b(v,C)-1$\;
					\lIf{$\cnt_b(v,C)=0$}{
						$\remove_b(C)\ass\remove_b(C)\cup\{v\}$\hspace{-2pt}
					}
				}
			}
		}
        }
}
}

\newcommand{\TheAlgorithm}{
\begin{algorithm}[t]\label{algorithm}
\Algo
\end{algorithm}
}

\newcommand{\tableiii}{
\begin{table}[ht]
\centering{
\scalebox{0.9}{
\begin{tabular}{|l|rrrr|rrr|rr|rr|}
\hline
 & \multicolumn{4}{|c|}{\bf TA} & \multicolumn{3}{|c|}{\bf LTS} & \multicolumn{2}{|c|}{\bf LRT}  & \multicolumn{2}{|c|}{\bf OLRT}\\
\ccb<source> & \cc<$|Q|$> & \cc<$|\Sigma|$> & \cc<$\maxrank$> & \ccr<$|\Delta|$> & \cc<$|S|$> & \cc<$|\Sigma|$> & \ccr<$|\delta|$> & \cc<time> & \ccr<space> & \cc<time> & \ccr<space> \\
\hline
random  & 16 &  16 & 2 &  245 &  472 & 17 &   952 &   1.03 &    96.5 &  0.09 &   4.8\\
random  & 32 &  16 & 2 &  935 & 1791 & 17 &  3700 &  18.73 &  1253.8 &  1.37 &  54.7\\
random  & 64 &  16 & 2 & 3725 & 7126 & 17 & 14824 & 405.89 & 14173.9 & 22.83 & 752.6\\
\hline
random  & 32 &  32 & 2 & 1164 & 2204 & 33 &  4548 &  64.10 & 3786.7 &  2.36 & 193.4\\
random  & 32 &  64 & 2 & 2026 & 3787 & 65 &  7874 & \multicolumn{2}{c|}{o.o.m.} &  6.72 & 245.8\\
\hline
ARTMC\footnotemark[1] & 47 & 132 & 2 &  837 & 1095 & 133 & 3344 &  66.46 & 4183.2 &  0.69 &  68.2\\
\hline
ARTMC & \multicolumn{7}{c|}{variable\footnotemark[2]} & 12669.94  & 4412.6  & 400.62 & 106.6\\
\hline
\end{tabular}
}
\caption{Upward simulation results}
\label{res3}
}
\end{table}
}

\newcommand{\tableii}{
\begin{table}[ht]
\centering{
\scalebox{0.9}{
\begin{tabular}{|l|rrrr|rrr|rr|rr|}
\hline
 & \multicolumn{4}{|c|}{\bf TA} & \multicolumn{3}{|c|}{\bf LTS} & \multicolumn{2}{|c|}{\bf LRT}  & \multicolumn{2}{|c|}{\bf OLRT}\\
\ccb<source>  & \cc<$|Q|$> & \cc<$|\Sigma|$> & \cc<$\maxrank$> & \ccr<$|\Delta|$> & \cc<$|S|$> & \cc<$|\Sigma|$> & \ccr<$|\delta|$> & \cc<time> & \ccr<space> & \cc<time> & \ccr<space>\\
\hline
random  & 16 &  16 & 2 &  245 &  184 &  18 &  570 &   0.06  &    6.2 &  0.02 &   1.4\\
random  & 32 &  16 & 2 &  935 &  655 &  18 & 2165 &   0.87  &   74.4 &  0.21 &  14.4\\
random  & 64 &  16 & 2 & 3725 & 2502 &  18 & 8568 &  26.63  & 1417.9 &  3.50 & 195.4\\
\hline
random  & 32 &  32 & 2 & 1164 &  719 &  34 & 2511 &   2.67  &  166.6 &  0.23 &  16.8\\
random  & 32 &  64 & 2 & 2026 &  925 &  66 & 3780 &  12.17  &  623.5 &  0.56 &  25.4\\
\hline
ARTMC\footnotemark[1] & 47 & 132 & 2 &  837 &  241 & 134 & 1223 &   0.84  & 70.6 &  0.05 &   6.2\\
\hline
ARTMC & \multicolumn{7}{c|}{variable\footnotemark[2]} & 517.98  & 116.2  & 80.84 & 22.1\\
\hline
\end{tabular}
}
\caption{Downward simulation results}
\label{res2}
}
\end{table}
\footnotetext[1]{One of the automata selected from the ARTMC set.}
\footnotetext[2]{A set containing 10\,305 tree automata of variable size (up to 50 states
and up to 1000 transitions per automaton). The results show the total amount of time required
for the computation and the peak size of allocated memory.}
}

\newcommand{\tablei}{
\begin{table}[ht]
\centering{
\scalebox{0.9}{
\begin{tabular}{|l|rrr|rr|rr|}
\hline
 & \multicolumn{3}{c|}{\bf LTS} & \multicolumn{2}{c|}{\bf LRT}  & \multicolumn{2}{c|}{\bf OLRT}\\
\ccb<source> & \cc<$|S|$> & \cc<$|\Sigma|$> & \ccr<$|\delta|$> & \cc<time> & \ccr<space> & \cc<time> & \ccr<space>\\
\hline
random &   256 & 16 &   416 &   0.12 &     9.6 &   0.02 &    1.9\\
random &  4096 & 16 &  3280 &  13.82 &   714.2 &   2.02 &   78.2\\
random & 16384 & 16 & 26208 & \multicolumn{2}{c|}{o.o.m.} & 268.85 & 4514.9\\
\hline
random &  4096 & 32 &  6560 &  62.09 &  1844.2 &   4.36 &  121.4\\
random &  4096 & 64 & 13120 & 158.38 &  3763.2 &   6.59 &  211.2\\
\hline
pc     &  1251 & 43 & 49076 &   7.52 &   418.1 &   2.63 &  119.0\\
rw     &  4694 & 11 & 20452 &  81.28 &  3471.8 &  19.25 &  989.3\\
lr     &  6160 & 35 & 90808 & 390.91 & 12640.8 &  45.69 & 1533.6\\
\hline
\end{tabular}
}
\caption{LTS simulation results}
\label{res1}
}
\end{table}
}

\def\cc<#1>{\multicolumn{1}{c}{#1}}
\def\ccr<#1>{\multicolumn{1}{c|}{#1}}
\def\ccb<#1>{\multicolumn{1}{|c|}{#1}}

\begin{document}


\title{Optimizing an LTS-Simulation Algorithm\thanks{This is a version of Technical Report No. FIT-TR-2009-03 of Faculty of Information Technology, Brno University of Technology, accompanying the work published originally as \cite{FITPUB9733}. }}

\author{
  Luk\'{a}\v{s} Hol\'{\i}k$^1$ \and
  Ji\v r\' i \v Sim\'a\v cek$^{1,2}$
}

\institute{\scriptsize
  $^1$ FIT BUT, Bo\v{z}et\v{e}chova 2, 61266 Brno, Czech Republic \\
  $^2$ VERIMAG, UJF, 2. av. de Vignate, 38610 Gi\`{e}res, France \\
  email: \email{holik@fit.vutbr.cz}
}

\maketitle
\thispagestyle{plain}

\begin{abstract}
When comparing the fastest algorithm for computing the largest simulation preorder over
Kripke structures with the one for labeled transition systems (LTS), there is
a noticeable time and space complexity blow-up proportional to the size of the
alphabet of an LTS.  In this paper, we present optimizations that suppress this
increase of complexity and may turn a large alphabet of an LTS to an advantage. Our
experimental results show significant speed-ups and memory savings. Moreover,
the optimized algorithm allows one to improve asymptotic complexity of
procedures for computing simulations over tree automata using recently proposed
algorithms based on computing simulation over certain special LTS derived from a tree automaton. 
\end{abstract}

\section{Introduction}
%
A practical limitation of automated methods dealing with LTSs---such as
LTL model checking, regular model checking, etc.---is often the size of generated
LTSs. 
One of the well established approaches to overcome this problem is the reduction
of an LTS using a suitable equivalence relation according to which the
states of the LTS are collapsed.
A good candidate for such a relation is simulation equivalence. 
It strongly preserves logics like  $ACTL^*$, $ECTL^*$, and $LTL$ \cite{dams-generation,grumberg-model,henzinger-computing},
and with respect to its reduction power and computation cost, it offers a
desirable compromise 
among the other common candidates, such as bisimulation 
equivalence \cite{piage-three,sawa-behavioural} and language equivalence. 
%
The currently fastest LTS-simulation algorithm
(denoted as LRT---labeled RT) has been published in \cite{abdulla-computing}. 
It is a straightforward modification of the fastest algorithm (denoted as RT, standing for Ranzato-Tapparo) for computing
simulation over Kripke structures \cite{ranzato-new}, which improves the
algorithm from \cite{henzinger-computing}. 
The time complexity of RT amounts to $\O(|P_{\sim}||\delta|)$, space complexity
to  $\O(|P_{\sim}||S|)$.
In the case of LRT, we obtain time complexity $\O(|P_{\sim}||\delta| + |\labs||P_\sim||S|)$ and space complexity $\O(|\labs||P_{\sim}||S|)$. 
Here, $S$ is the set of states, $\delta$ is the transition relation, $\labs$ is the
alphabet and $P_\sim$ is the partition of $S$ according to the simulation
equivalence.
%
The space complexity blow-up of LRT is caused by indexing the data structures of RT
by the symbols of the alphabet.

In this paper, we propose an optimized version of LRT (denoted OLRT) that lowers the above described blow-up.
We exploit the fact that not all states of an LTS have incoming and outgoing
transitions labeled by all symbols of an alphabet, which allows us to reduce
the memory footprint of the data structures used during the computation. 
Our experiments show that the optimizations we propose lead to significant
savings of space as well as of time in many practical cases.
Moreover, we have achieved a promising reduction of the asymptotic complexity of
algorithms for computing tree-automata simulations from
\cite{abdulla-computing} using OLRT.
%

\section{Preliminaries}
Given a binary relation $\rho$ over a set $X$, 
we use $\above \rho x$ to denote the set $\{y\mid (x,y)\in \rho\}$. 
Then, for a set $Y\subseteq X$, $\rho(Y) = \bigcup\{\rho(y)\mid y\in Y\}$.
A \emph{partition-relation pair} over $X$ is a pair $\lara{P,Rel}$ where
$P \subseteq 2^X$ is a partition of $X$ (we call elements of $P$ \emph{blocks})
and $Rel \subseteq P \times P$. A partition-relation pair $\lara{P,Rel}$
\emph{induces} the relation ${\rho} = {\bigcup_{(B,C)\in
Rel} B \times C}$. We say that $\lara{P,\rel}$ is the \emph{coarsest partition-relation pair} inducing 
$\rho$ if any two $x,y\in X$ are in the same block of $P$ if and only if
$\above \rho x = \above \rho y$ and $\bellow \rho x = \bellow \rho y$. 
Note that in the case when $\rho$ is a preorder and
$\lara{P,\rel}$ is coarsest, then $P$ is the set of equivalence classes of
$\rho\cap\rho^{-1}$ and $\rel$ is a partial order. 

A \emph{labeled transition system (LTS)} is a tuple
\mbox{$T=(S,\labs,\trrelset)$}, where $S$ is a finite set of
states, $\labs$ is a finite set of labels, and for each $a\in\labs$,
$\tr_a\subseteq{ S\times S}$ is an $a$-labeled transition relation. We use $\delta$ to denote
$\bigcup_{a\in\labs}\delta_a$.
%
%
A {\it simulation} over $T$ is a binary relation
$\wordsimrel$ on $S$ such that if $(u,v)\in\wordsimrel$, then for all
$a\in\labs$ and $u'\in\post_a(u)$, there exists $v'\in\post_a(v)$ such that
$(u',v')\in\wordsimrel$. 
It can be shown that for a given LTS $T$ and an {\it initial
preorder} $\init \subseteq S \times S$,  there is a unique maximal simulation
$\wordsim$ on $T$ that is a subset of $\init$, and that $\wordsim$ is a preorder (see~\cite{abdulla-computing}).



\section{The Original LRT Algorithm}
\label{oltsa}
In this section, we describe the original version of the algorithm presented in
\cite{abdulla-computing}, which we denote as LRT (see Algortihm \ref{algorithm}).

\TheAlgorithm

The algorithm  gradually refines a partition-rela\-tion pair $\lara{P,\rel}$,
which is initialized as the coarsest partition-relation pair inducing an
initial preorder $\init$. After its termination, $\lara{P,\rel}$ is the
coarsest partition-relation pair inducing $\wordsim$. The basic invariant of
the algorithm is that the relation induced by $\lara{P,\rel}$ is always a
superset of $\wordsim$.

The while-loop refines the partition $P$ and then prunes the relation $\rel$
in each iteration of the while-loop.
The role of the $\remove$ sets can be explained as follows:
During the initialization, every $\remove_a(B)$ is filled by states $v$ such
that $\post_a(v)\cap\bigcup\rel(B) = \emptyset$ 
(there is no $a$-transition leading from $v$ ``above'' $B$
wrt. $\rel$). During the computation phase, $v$ is added into $\remove_a(B)$
after $\post_a(v)\cap\bigcup\rel(B)$ becomes empty (because of pruning $\rel$
on line 17). 
Emptiness of
$\post_a(v)\cap \rel(B)$ is tested on line 20 using counters $\cnt_a(v,B)$,
which record the cardinality of $\post_a(v)\cap\rel(B)$.
From the definition of simulation, and because the
relation induced by $\lara{P,\rel}$ is always a superset of $\wordsim$,
$\post_a(v)\cap\bigcup\rel(B) = \emptyset$
implies that for all $u\in\pre_a(B)$, $(u,v)\not\in\wordsim$ ($v$ cannot
simulate any $u\in\pre_a(B)$).
To reflect this, the relation $\rel$ is pruned each time $\remove_a(B)$ is
processed.
The code on lines 8--13 prepares the
partition-relation pair and all the data structures. First,
$\Split(P,\remove_a(B))$ divides every block $B'$ into $B'\cap\remove_a(B)$
(which cannot simulate states from $\pre_a(B)$ as they have empty intersection
with $\pre_a(\rel(B))$), and $B'\setminus\remove_a(B)$.
More specifically, for a set $\remove\subseteq S$,
$\Split(P,\remove)$ returns a finer partition $P' = \{B\setminus \remove\mid
B\in P\}\cup \{B\cap \remove\mid B\in P\}$. After refining $P$ by the $\Split$
operation, the newly created blocks of $P$ inherit the data structures
(counters $\cnt$ and $\remove$ sets) from  their ``parents'' (for a block $B\in
P$, its parent is the block $B_\prev\in P_\prev$ such that $B\subseteq
B_\prev$). 
$\rel$ is then updated on line 17 by removing the pairs $(C,D)$ such that
$C\cap \pre_a(B)\neq\emptyset$ and $D\subseteq\remove_a(B)$. 
The change of $\rel$ causes that
for some states $u\in S$ and symbols $b\in\labs$,
$\post_a(u)\cap\bigcup\rel(C)$ becomes empty. To propagate the change of the relation along the transition relation, $u$ will be moved into
$\remove_b(C)$ on line 20, which will cause new changes of the relation in the following iterations of the while-loop.
If there is no nonempty $\remove$ set, then
$\lara{P,\rel}$ is the coarsest partition-relation pair inducing $\wordsim$
and the algorithm terminates. 
Correctness of LRT is stated by Theorem \ref{theorem-LRT}.
\begin{theorem}[\cite{abdulla-computing}]\label{theorem-LRT}
With an LTS $T=(S,\labs,\trrelset)$ and the coarsest partition-relation pair
$\lara{P_\init,\rel_\init}$ inducing a preorder $I\subseteq S\times S$
on the input, LRT terminates with the coarsest partition-relation pair
$\lara{P,\rel}$ inducing $\wordsim$.
\end{theorem}

\section{Optimizations of LRT}
The optimization we are now going to propose reduces the number of counters and the number and the size of
$\remove$ sets. The changes required by OLRT are indicated  in Algorithm \ref{algorithm} on the right hand
sides of the concerned lines.

We will need the following notation. For a state $v\in S$, $\inp v =
\{a\in\labs\mid\pre_a(v)\neq\emptyset\}$ is the set of \emph{input symbols} and
$\out v = \{a\in\labs\mid\post_a(v)\neq\emptyset\}$ is the set of \emph{output
symbols} of $v$.  The \emph{output preorder} is the relation $\outpre =
\bigcap_{a\in\labs} \pre_a(S)\times\pre_a(S)$ 
(this is, ${(u,v)\in\outpre}$ if and only if ${\out u\subseteq\out v}$).

To make our optimization possible, we have to initialize $\lara{P,\rel}$ by the
finer partition-relation pair
$\lara{P_{\init\cap\outpre},\rel_{\init\cap\outpre}}$ (instead of
$\lara{P_{\init},\rel_{\init}}$), which is the coarsest partition-relation pair
inducing the relation $\init\cap\outpre$.  As both $\init$ and $\outpre$ are
preorders, $\init\cap\outpre$ is a preorder too. As $\wordsim\subseteq\init$
and $\wordsim\subseteq\outpre$ (any simulation on $T$ is a subset of
$\outpre$), $\wordsim$ equals the maximal simulation included in
$\init\cap\outpre$. Thus, this step itself does not influence the output of the
algorithm.

Assuming that $\lara{P,\rel}$ is initialized to
$\lara{P_{\init\cap\outpre},\rel_{\init\cap\outpre}}$, we can
observe that for any $B\in P$ and $a\in\labs$ chosen on line 5, the following two claims hold:

\begin{claim}
If $a\not\in\inp B$, then skipping this iteration of the while-loop does not
affect the output of the algorithm.
\end{claim}
\begin{proof}
In an iteration of the while-loop processing $\remove_a(B)$ with $a\not\in\inp
B$, as there is no $C\in P$ with $\delta_a(C)\cap \above \rel B \neq
\emptyset$, the for-loop on line 16 stops immediately. No pair $(C,D)$ will be
removed from $\rel$ on line 17, no counter will be decremented, and no state
will be added into a $\remove$ set. The only thing that can happen is that
$\Split(P,\remove)$ refines $P$. However, in this case, this refinement of $P$
would be done anyway in other iterations of the while-loop when processing sets
$\remove_b(C)$ with $b\in\inp C$. To see this, note that correctness of the
algorithm does not depend on the order in which nonempty $\remove$ sets are
processed. Therefore, we can postpone processing all the nonempty
$\remove_a(B)$ sets with $a\not\in\inp B$ to the end of the computation.  Recall
that processing no of these $\remove$ sets can cause that an empty $\remove$
set becomes nonempty. Thus, the algorithm terminates after processing the last
of the postponed $\remove_a(B)$ sets. If processing some of these
$\remove_a(B)$ with $a\not\in\inp B$ refines $P$, $P$ will contain blocks $C,D$
such that both $(C,D)$ and $(D,C)$ are in $\rel$ (recall that when processing
$\remove_a(B)$, no pair of blocks can be removed from $\rel$ on line 17). This
means that the final $\lara{P,\rel}$ will not be coarsest, which is a
contradiction with Theorem \ref{theorem-LRT}. Thus, processing the postponed
$\remove_a(B)$ sets can influence nor $\rel$ neither $P$, and therefore they do not have to be processed at all. 
\end{proof}

\begin{claim}
It does not matter whether we assign $\remove_a(B)$ or $\remove_a(B)\setminus (S\setminus \pre_a(S))$ to $\remove$ on line 6.
\end{claim}
\begin{proof}
Observe that $v$ with $a\not\in\out v$ (i.e., $v\in S\setminus\pre_a(S)$)
cannot be added into $\remove_a(B)$ on line 20, as this would mean that $v$ has
an $a$-transition leading to $D$. Therefore, $v$ can get into $\remove_a(B)$
only during initialization on line 4 together with all states from
$S\setminus\pre_a(S)$. After $\remove_a(B)$ is processed (and emptied) for the
first time, no state from $S\setminus\pre_a(S)$ can appear there again.  Thus,
$\remove_a(B)$ contains states from $S\setminus\pre_a(S)$ only when it is
processed for the first time and then it contains all of them. It can be shown
that for any partition $Q$ of a set $X$ and any $Y\subseteq X$, if $\Split(Q,Y)
= Q$, then also for any $Z\subseteq X$ with  $Y\subseteq Z$, $\Split(Q,Z) =
\Split(Q,Z\setminus Y)$.  As $P$ refines $P_{\init\cap\outpre}$, $\Split(P,S
\setminus\pre_a(S)) = P$. Therefore, as $S\setminus\pre_a(S) \subseteq
\remove_a(B)$, $\Split(P,\remove_a(B)) =
\Split(P,\remove_a(B)\setminus(S\setminus\pre_a(S)))$. We have shown that
removing $S\setminus\pre_a(S)$ from $\remove$ does not influence the result of
the $\Split$ operation in this iteration of the while-loop (note that this
implies that all blocks from the new partition are included in or have empty
intersection with $S\setminus\pre_a(S)$). 
It remains to show that it also does not influence updating $\rel$ on line 17.
Removing $S\setminus\pre_a(S)$ from $\remove$ could only cause that the blocks
$D$ such that $D\subseteq S\setminus\pre_a(S)$ that were chosen on line 15 with
the original value of $\remove$ will not be chosen with the restricted
$\remove$. Thus, some of the pairs $(C,D)$ removed from $\rel$ with the
original version of $\remove$ could stay in $\rel$ with the restricted version
of $\remove$. However, such a pair $(C,D)$ cannot exist because with the
original value of $\remove$, if $(C,D)$ is removed from $\rel$, then $a\in\out
C$ (as $\delta(C)\cap B\neq\emptyset$) and therefore also $a\in\out D$ (as
$\rel$ was initialized to $\rel_{\init\cap\outpre}$ on line 1 and
$(C,D)\in\rel$). Thus, 
$D\cap (S\setminus\pre_a(S))=\emptyset$, which means that $(C,D)$ is removed
from $\rel$ even with the restricted $\remove$. Therefore, it does not matter
whether $S\setminus\pre_a(S)$ is a subset of or it has an empty intersection
with $\remove$.  \end{proof}


As justified above, we can optimize LRT as follows.  Sets
$\remove_a(B)$ are computed only if $a \in \inp B$ and in that case we only add
states $q \in \pre_a(S)$ to $\remove_a(B)$. As a result, we can reduce the
number of required counters by maintaining $\cnt_a(v,B)$ if and only if
$a\in\inp B$ and $a\in\out v$.

\section{Implementation and Complexity of OLRT}

We first briefly describe the essential data structures (there are some additional data structures required by our optimizations) and then we sketch the
complexity analysis.  For the full details, see the technical report
\cite{techrep}.

\paragraph{Data Structures.}
The input LTS is represented as a list of records about its states. The record
about each state $v \in S$ contains a list of nonempty $\pre_a(v)$
sets\footnote{We use a list rather than an array having an entry for each $a
\in \labs$ in order to avoid a need to iterate over alphabet symbols for which
there is no transition.}, each of them encoded as a list of its members.  The
partition $P$ is encoded as a doubly-linked list (DLL) of blocks. Each block is
represented as a DLL of (pointers to) states of the block.  Each block $B$
contains for  each $a\in\labs$ a list of (pointers on) states from
$\remove_a(B)$. Each time when any set $\remove_a(B)$ becomes nonempty, block
$B$ is moved to the beginning of the list of blocks. Choosing the block $B$ on
line 5 then means just scanning the head of the list of blocks.  The relation
$\rel$ is encoded as a resizable boolean matrix.

Each block $B\in P$ and each state $v\in S$ contains an $\labs$-indexed array
containing a record $B.a$ and $v.a$, respectively.  The record $B.a$ stores the
information whether $a\in\inp B$ (we need the test on $a\in\inp B$ to take a
constant time), 
If $a\in\inp B$, then $B.a$ also contains a reference to the set
$\remove_a(B)$, represented as a list of states (with a constant time
addition), and a reference to an array of counters $B.a.\cnt$ containing the
counter $\cnt_a(v,B)$ for each $v\in\pre_a(S)$. 
Note that for two different symbols $a,b\in\labs$ and some $v\in S$, the
counter $\cnt_a(v,B)$ has different index in the array $B.a.\cnt$ than the
counter $\cnt_b(v,B)$ in $B.b.\cnt$ (as the sets $\pre_a(S)$ and $\pre_b(S)$
are different).  Therefore, for each $v\in S$ and $a\in\labs$,  $v.a$ contains
an index $v_a$ under which for each $B\in P$, the counter $\cnt_a(v,B)$ can be
found in the array $B.a.\cnt$. Using the $\labs$-indexed arrays attached to
symbols and blocks, every counter can be found/updated in a constant time.  For
every $v\in S,a\in\labs$, $v.a$ also stores a pointer to the list containing
$\pre_a(v)$ or $null$ if $\pre_a(v)$ is empty.  This allows the constant time
testing whether $a\in\inp v$ and the constant time searching for the
$\pre_a(v)$ list. 

\paragraph{Complexity analysis (Sketch).}
We first point out how our optimizations influence complexity of the most costly part of the code which is the main while loop.
The analysis of lines 14--16 of LRT is based on the observation that for any two
$B',D'\in P_\wordsim$ and any $a\in\labs$, it can happen at most once that $a$
and some $B$ with $B'\subseteq B$ are chosen on line 14 and at the same time
$D'\subseteq \remove_a(B)$.  In one single iteration of the while-loop, blocks $C$
are listed by traversing all $\pre(v),v\in B$ (the $D$s can be enumerated
during the $\Split$ operation).
%
Within the whole computation, for any $B'\in  P_\wordsim$, transitions leading
to $B'$ are traversed on line 14 at most $P_\wordsim$ times, so the complexity
of lines 14--16 of LRT is $\O(\sum_{a\in\labs}\sum_{D\in
P_\sim}\sum_{v\in S}|\pre_a(v)|) = \O(|P_\wordsim||\delta|)$.
%
%
In the case of OLRT, the number and the content of remove sets is
restricted in such a way that for a nonempty set $\remove_a(B)$, it holds that $a\in\inp B$ and
$\remove_a(B)\subseteq \pre_a(S)$. Hence, for a fixed $a$, $a$-transition
leading to a block $B'\in P_\wordsim$ can be traversed only $|\{D'\in
P_\wordsim\mid a\in\out {D'}\}|$ times and the complexity of lines 14--16
decreases to $\O(\sum_{D\in P_\wordsim}\sum_{a\in\inp D}|\delta_a|)$.

The analysis of lines 17--20 of LRT is based on the fact that once $(C,D)$ appears on
line 17, no $(C',D')$ with $C'\subseteq C, D'\subseteq D$ can appear there
again. For a fixed $(C,D)$, the time spent on lines 17--20 is in $\O(\sum_{v\in
B}|\pre(v)|)$ and only those blocks $C,D$ can meet on line 17 such that
$C\times D\subseteq \init$. Thus, the overall time spent by LRT on lines 17--20
is in $\O(\sum_{B\in P_\wordsim}\sum_{v\in\init(B)}|\pre(v)|)$.  In OLRT,
blocks $C,D$ can meet on line 17 only if $C\times D\subseteq\init\cap\outpre$,
and the complexity of lines 17--20 in OLRT decreases to $\O(\sum_{B\in
P_\wordsim}\sum_{v\in(\init\cap\outpre)(B)}|\pre(v)|)$.

Additionally, OLRT refines $\lara{P_\init,\rel_\init}$ to
$\lara{P_{\init\cap\outpre},\rel_{\init\cap\outpre}}$ on line 1. This can be
done by successive splitting according to the sets $\pre_a(S),a\in\labs$ and
after each split, breaking the relation between blocks included in $\pre_a(S)$
and the ones outside. 
This procedure takes time $\O(|\labs||P_{\init\cap\outpre}|^2)$.

Apart from some other smaller differences, the implementation and the
complexity analysis of OLRT are analogous to the implementation and the
analysis of LRT \cite{abdulla-computing}.
%
%
%
%
%
The overall time complexity of OLRT is 
$
\O\bigl( |\labs||P_{\init\cap\outpre}|^2 +  |\labs||S|+  |P_{\wordsim}|^2 + 
\sum_{B\in P_\wordsim}
(\sum_{a\in\inp B} (|\pre_a(S)| + |\delta_a|)+
\sum_{v\in(\init\cap\outpre)(B)}|\pre(v)|)\bigr).
$

The space complexity of OLRT is determined by the number of counters, the contents of
the $\remove$ sets, the size of the matrix encoding of $\rel$, and the space needed for storing the $B.a$ and $v.a$ records (for every block $B$, state $v$ and symbol $a$).  
Overall, it gives $\O(|P_\wordsim|^2 + |\Sigma||S| + \sum_{B\in
P_\wordsim}\sum_{a\in\inp B}|\delta_a^{-1}(S)|)$.  

Observe that the improvement of both time and space complexity of LRT 
is most significant for systems with large
alphabets and a high diversity of sets of input and output symbols of states.
Certain regular diversity of sets of input and output
symbols is an inherent property of LTSs that arise when we compute simulations over
tree automata. We address the impact of employing OLRT within the procedures for computing tree automata simulation in the next section.

\section{Tree Automata Simulations}\label{tass}
In \cite{abdulla-computing}, authors  propose methods for computing tree
automata simulations via translating problems of computing simulations over
tree-automata to problems of computing simulations over certain LTSs. In this
section, we show how replacing LRT by OLRT within these translation-based
procedures decreases the overall complexity of computing tree-automata simulations. 

A (finite, bottom-up) \emph{tree automaton} (TA) is a quadruple
$A = (Q, \Sigma,\Delta, F)$  where $Q$ is a finite
set of states, $F \subseteq Q$ is a set of final states, $\Sigma$ is a ranked
alphabet with a ranking function $\rankof:\Sigma \rightarrow \nat$, and
$\Delta\subseteq Q^*\times\Sigma\times Q$ is a set of transition rules such
that if $(q_1\ldots q_n,f,q)\in\Delta$, then $\rankof(f) = n$.  Finally, we
denote by $\maxrank\Aref$ the smallest $n \in \nat$ such  that $n \geq m$ for
each $m \in \nat$ such that there is some $(q_1\ldots q_m,f,q) \in\Delta$.  
We omit the definition of the semantics of TA as we will not need it,
and we only refer to \cite{tata,abdulla-computing}.

For the rest of this section, we fix a TA $A = (Q,\Sigma,\Delta,F)$. A
\emph{downward simulation} $D$ is a binary relation on $Q$ such that if
$(q,r)\in D$, then for all $\trans q n f q\in\Delta$, there exists $\trans r n
f r\in\Delta$ such that $(q_i,r_i)\in D$ for each $i:1\leq i\leq n$.  Given a
downward simulation $\lpre$ which is a preorder called an \emph{inducing
simulation}, an \emph{upward simulation $\us$ induced by $\lpre$} is a binary
relation on $Q$ such that if $(q,r)\in\us$, then (i) for all $\trans q n f
{q'}\in\Delta$ with $q_i=q,1 \leq i \leq n$, there exists $\trans r n f
{r'}\in\Delta$ with $r_i=r$, $(q',r')\in \us$, and $(q_j,r_j) \in \lpre$ for
each $j:1\leq j\neq i\leq n$; (ii) $q\in F\implies r\in F$.
From now on, let $\downsim$ denote the maximal downward simulation on $A$ and
$\upsim$ the maximal upward simulation on $A$ induced by $\downsim$.


To define the translations from downward and upward simulation problems, we
need the following notions. Given a transition $t = (q_1\ldots
q_n,f,q)\in\Delta$,  $q_1\ldots q_n$ is its \emph{left-hand side} and $t(i)\in
(Q\cup\{\hole\})^*\times \Sigma\times Q$ is an \emph{environment}---the tuple
which arises from $t$ by replacing state $q_i$, $1 \leq i \leq n$, at the
$i^{th}$ position of the left-hand side~of $t$ by the so called hole $\hole
\not\in Q$. We use $\lhsof$ of to denote the set of all left-hand sides of $A$
and $\envof$ to denote the set of all environments~of~$A$.


We translate the downward simulation problem on $A$ to the simulation problem
on the LTS  $A^\bullet = (Q^\bullet,\labs^\bullet,\{\delta_a^\bullet\mid a\in\labs^\bullet\})$
where $Q^\bullet = \{q^\bullet\mid q\in Q\}\cup\{l^\bullet\mid l\in\lhsof\Aref\}$, $\Sigma^\bullet = \Sigma\cup
\{1,\ldots, \maxrank\}\}$, 
and for each $(q_1\ldots q_n,f,q)\in\Delta$,  $(q^\bullet,
q_1\ldots q_n^\bullet)\in\delta^\bullet_f$ and $(q_1\ldots q_n^\bullet,
q_i^\bullet)\in\delta^\bullet_i$ for each  $i:1\leq i \leq n$. 
The initial relation is simply $\init^\bullet = Q^\bullet\times
Q^\bullet$. 
The upward simulation problem is then translated into a simulation problem on
LTS $A^\odot = (Q^\odot,\Sigma^\odot,\{\delta^\odot_a\mid a\in\Sigma^\odot\})$,
where $Q^\odot = \{q^\odot\mid q\in Q\}\cup\{e^\odot \mid e\in\envof\Aref\}$,
$\Sigma^\odot = \Sigma^\bullet$, and for each $t=(q_1\ldots q_n,f,q)\in\Delta$, for each
$1\leq i\leq n$, $(q_i^\odot,t(i)^\odot)\in\delta^\odot_{i}$ and $(t(i)^\odot,q^\odot)\in\delta^\odot_a$.
The initial relation $I^\odot \subseteq
Q^\odot\times Q^\odot$ contains all the pairs
$(q^\odot,r^\odot)$  such that $q,r\in Q$ and $r\in F\implies q\in F$, and
$((q_1\ldots q_n,f,q)(i)^\odot,(r_1 \ldots r_n,f,r)(i)^\odot)$ such that
$(q_j,r_j)\in\downsim$ for all $j:1\leq i\neq j \leq n$. Let $\sim^\bullet$ be
the maximal simulation on $A^\bullet$ included in $\init^\bullet$ and let
$\sim^\odot$ be the maximal simulation on $A^\odot$ included in $\init^\odot$.
The following theorem shows correctness of the translations.
\begin{theorem}[\cite{abdulla-computing}] \label{tree:sim:theorem} For all
$q,r\in Q$, we have $(q^\bullet,r^\bullet) \in\sim^\bullet$ if and only if $(q,r)\in
\downsim$ and $(q^\odot,r^\odot) \in\sim^\odot$ if and only if $(q,r)\in
\upsim$. \end{theorem}

The states of the LTSs ($A^\bullet$ as well as $A^\odot$) can be classified
into several classes according to the sets of input/output
symbols. Particularly, $Q^\bullet$ can be classified into the
classes $\{q^\bullet\mid q\in Q\}$ and for each $n:1\leq n\leq\maxrank\Aref$,
\{$q_1\ldots q_n^\bullet\mid q_1\ldots q_n\in\lhsof\Aref\}$, and $Q^\odot$ can
be classified into $\{q^\odot\mid q\in Q\}$ and for each $a\in\Sigma$ and
$i:1\leq i\leq \rankof(a)$, $\{t(i)^\odot\mid t = (q_1\ldots
q_n,a,q)\in\Delta\}$. This turns to a significant advantage when computing
simulations on $A^\bullet$ or on $A^\odot$ using OLRT instead of LRT. Moreover, we now propose another
small optimization, which is a specialized procedure for
computing $\lara{P_{\init\cap\outpre}\rel_{\init\cap\outpre}}$ for the both of
$A^\odot$, $A^\bullet$.  It is based on the simple observation that we need
only a constant time (not a time proportional to the size of the alphabet) to determine
whether two left-hand sides or two environments are related by the particular $\outpre$ (more
specifically, $(e_1^\odot,e_2^\odot)\in\outpre$ if and only if the inner symbols of
$e_1$ and $e_2$ are the same, and $(q_1\ldots
q_n^\bullet,r_1\ldots r_m^\bullet)\in\outpre$ if and only if $n\leq m$).

\paragraph{Complexity of the Optimized Algorithm.} 
We only point out the main differences between application of 
LRT \cite{abdulla-computing} and OLRT on the LTSs that arise from the translations described above.  For implementation details and full complexity
analysis of the OLRT versions, see the technical report \cite{techrep}.

To be able to express the complexity of running OLRT on $A^\bullet$
and $A^\odot$, we extend $\downsim$ to the set $\lhsof\Aref$ such that
$(\vect q n,\vect r n)\in \downsim$ if and only if  $(q_i,r_i)\in\downsim$ for each
$i:1\leq i\leq n$, and we extend $\upsim$ to the set $\envof\Aref$ such that
\\$((q_1 \dotsc q_n,f,q)(i),(r_1 \dotsc r_n,f,r)(i))
\in\upsim {\iff} m=n \wedge i=j \wedge (q,r)\in\upsim
\wedge (\forall k \in \{ 1, ..., n\}.\ k\neq i \Longrightarrow \
(q_k,r_k)\in\downsim)$.  For a preorder $\rho$ over a set $X$, we use
$\prt{X}{\rho}$ to denote the partition of $X$ according to the equivalence
$\rho\cap\rho^{-1}$. 

The procedures for computing $\sim^\bullet$ and $\sim^\odot$ consist of (i)
translating $A$ to the particular LTS ($A^\bullet$ or $A^\odot$) and computing
the partition-relation pair inducing the initial preorder ($\init^\bullet$ or
$\init^\odot$), and (ii) running a simulation algorithm (LRT or OLRT) on it.
Here, we analyze the impact of replacing LRT by OLRT on the complexity of step
(ii), which is the step with dominating complexity (as shown in
\cite{abdulla-computing} and also by our experiments; step (ii) is much more
computationally demanding than step (i)). 

As shown in the technical report  \cite{techrep}, OLRT takes on $A^\bullet$ and $\init^\bullet$ space
$\O(Space_D)$ where  
$Space_D {=} 
(\maxrank + |\labs|)|\lhsof\cup Q| +
|\lhsof\cup Q/\downsim|^2 +
|\labs||\lhsof/\downsim||Q| + \maxrank|Q/\downsim||\lhsof|
$ and time
$\O(Space_D +
|\labs||Q/\downsim|^2 +
|\lhsof/\downsim||\Delta|
)$.
On $A^\odot$ and $\init^\odot$, OLRT runs in 
time $O(Space_U)$ where
$Space_U = (\maxrank + |\labs|)|\envof| +
|\envof/U|^2 + |\envof/U||Q| + |Q/U||\envof|$ and time $\O(Space_U +
|\labs||Q/U|^2 + |\envof/U||\delta|)$. 

We compare the above results with \cite{abdulla-computing}, where LRT is used.
LRT on $A^\bullet$ and $\init^\bullet$ takes $\O(Space_D^{old})$ space where
$Space_D^{old} = (|\Sigma|+ \maxrank)|Q\cup\lhsof||\prt{Q\cup\lhsof}\downsim|$,
and $\O(Space_D^{old} + |\Delta||\prt{Q\cup\lhsof}{\downsim}|)$ time.  In the
case of $A^\odot$ and $I^\odot$, we obtain space complexity $\O(Space_U^{old})$ where
$Space_U^{old} = |\Sigma||\envof||\prt{\envof}\upsim|$ and time complexity
$\O(Space_U^{old} + \maxrank|\Delta||\prt{\envof}\upsim|)$.

The biggest difference is in the space complexity (decreasing
the factors $Space_D^{old}$ and $Space_U^{old}$).  However, the time complexity is
better too, and our experiments show a significant improvement in space as well as in time.

%
%
%
%
%

\section{Experiments}

We implemented the original and the improved version of the algorithm in a
uniform way in OCaml and experimentally compared their performance.

\tablei


The simulation algorithms were benchmarked using LTSs obtained from the runs of the abstract regular model checking (ARMC)
(see~\cite{vojnar05,vojnar04}) on several classic examples---producer-consumer (pc), readers-writers
(rw), and list reversal (lr)---and using a set of tree automata obtained from the run of the abstract regular tree model checking (ARTMC)
(see~\cite{antichain}) on several operations, such as list reversal,
red-black tree balancing, etc. We also used several randomly generated LTSs and tree automata.

\tableii


We performed the experiments on AMD Opteron 8389 2.90 GHz PC with 128 GiB of memory (however we set the
memory limit to approximately 20 GiB for each process). The system was running Linux and OCaml 3.10.2. 

The performance of the algorithms is compared in Table \ref{res1} (general LTSs), Table \ref{res2}
(LTSs generated while computing the downward simulation), and Table \ref{res3} (LTSs generated while
computing the upward
simulation), which contain the running times ([s]) and the amount of memory ([MiB]) required to finish
the computation.

\tableiii

As seen from the results of our experiments, our optimized implementation performs substantially better than the original.
On average, it improves the running time and space requirements by about one order of magnitude. As expected,
we can see the biggest improvements especially in the cases, where we tested the impact of the growing
size of the alphabet.

\section{Conclusion}
We proposed an optimized algorithm for computing simulations over LTSs,
which improves the asymptotic complexity in both space and time of the
best algorithm (LRT) known to date (see~\cite{abdulla-computing}) and which also
performs significantly better in practice. We also show how employing OLRT
instead of LRT reduces the complexity of the procedures for computing tree-automata
simulations from \cite{abdulla-computing}.
%
%
As our future work, we want to develop further optimizations, which would allow
to handle even bigger LTSs and tree automata. One of the possibilities is to
replace existing data structures by a symbolic representation, for example, by
using BDDs.

{
\medskip\noindent
{\bfseries Acknowledgements. }
This work was supported in part by the Czech Science Foundation (projects
P103/10/0306, 102/09/H042, 201/09/P531), 
the BUT FIT grant FIT-10-1,
the Czech COST project OC10009
associated with the ESF COST action IC0901, the Czech Ministry of Education by
the project MSM 0021630528, and the ESF project Games for Design and
Verification.

}

\bibliography{literature}

\begin{thebibliography}{10}

\bibitem{abdulla-computing}
P.A. Abdulla, A.~Bouajjani, L.~Hol\'{\i}k, L.~Kaati, and T.~Vojnar.
\newblock {Computing Simulations over Tree Automata: Efficient Techniques for
  Reducing Tree Automata}.
\newblock In {\em Proc. of TACAS'08}, LNCS~4963. Springer, 2008.

\bibitem{antichain}
A.~Bouajjani, P.~Habermehl, L.~Hol\'{i}k, T.~Touili, and T.~Vojnar.
\newblock {Antichain-Based Universality and Inclusion Testing over
  Nondeterministic Finite Tree Automata}.
\newblock In {\em Proc. of CIAA'08}, LNCS~5148. Springer, 2008.

\bibitem{vojnar05}
A.~Bouajjani, P.~Habermehl, P.~Moro, and T.~Vojnar.
\newblock {Verifying Programs with Dynamic 1-Selector-Linked Structures in
  Regular Model Checking}.
\newblock In {\em Proc. of TACAS'05}, LNCS~3440. Springer, 2005.

\bibitem{vojnar04}
A.~Bouajjani, P.~Habermehl, and T.~Vojnar.
\newblock {Abstract Regular Model Checking}.
\newblock In {\em Proc. of CAV'04}, LNCS~3114. Springer, 2004.

\bibitem{tata}
H.~Comon, M.~Dauchet, R.~Gilleron, C.~L\"oding, F.~Jacquemard, D.~Lugiez,
  S.~Tison, and M.~Tommasi.
\newblock {Tree Automata Techniques and Applications}.
\newblock \\\url{http://www.grappa.univ-lille3.fr/tata}, 2007.
\newblock release October, 12th 2007.

\bibitem{dams-generation}
D.~Dams, O.~Grumberg, and R.~Gerth.
\newblock {Generation of Reduced Models for Checking Fragments of {CTL}}.
\newblock In {\em Proc. of CAV'93}, 1993.

\bibitem{grumberg-model}
O.~Grumberg and D.~E. Long.
\newblock {Model Checking and Modular Verification}.
\newblock {\em ACM Transactions on Programming Languages and Systems}, 16,
  1994.

\bibitem{henzinger-computing}
M.~R. Henzinger, T.~A. Henzinger, and P.~W. Kopke.
\newblock {Computing simulations on finite and infinite graphs}.
\newblock In {\em Proc. of FOCS'95}. IEEE Computer Society, 1995.

\bibitem{techrep}
L.~Hol\'{\i}k and J.~\v{S}im\'{a}\v{c}ek.
\newblock {Optimizing an LTS-Simulation Algorithm}.
\newblock Technical Report FIT-TR-2009-03, Brno University of Technology, 2009.
\newblock \\\url{http://www.fit.vutbr.cz/~holik/pub/FIT-TR-2009-03.pdf}.

\bibitem{FITPUB9733}
Luk\'{a}\v{s} Hol\'{i}k and Ji\v{r}\'{i} \v{S}im\'{a}\v{c}ek.
\newblock Optimizing an lts-simulation algorithm.
\newblock {\em Computing and Informatics}, 2010(7):1337--1348, 2010.

\bibitem{piage-three}
R.~Piage and R.~Tarjan.
\newblock {Three Partition Refinement Algorithms}.
\newblock {\em SIAM Journal on Computing}, 16, 1987.

\bibitem{ranzato-new}
F.~Ranzato and F.~Tapparo.
\newblock {A New Efficient Simulation Equivalence Algorithm}.
\newblock In {\em Proc. of LICS'07}, 2007.

\bibitem{sawa-behavioural}
Z.~Sawa and P.~Jan\v{c}ar.
\newblock {Behavioural Equivalences on Finite-State Systems are PTIME-hard}.
\newblock {\em Computing and Informatics}, 24, 2005.

\end{thebibliography}


\end{document}